%% file: full.tex
\documentclass{llncs}
\usepackage{amsmath,amssymb,amsthm}
\usepackage{times}
\usepackage{url}
\usepackage{epsfig}
\usepackage[latin1]{inputenc}
\usepackage[mathscr]{euscript}
\usepackage{subfigure}
\usepackage{wrapfig}

\newenvironment{theo}[1][Theorem]{\begin{trivlist}
\item[\hskip \labelsep {\bfseries #1}]}{\end{trivlist}}

\newenvironment{prop}[1][Proposition]{\begin{trivlist}
\item[\hskip \labelsep {\bfseries #1}]}{\end{trivlist}}

\input xy \xyoption{all}
\input{macros}

\allowdisplaybreaks

\begin{document}

\title{Differential Privacy: on the trade-off between Utility and Information Leakage\thanks{This work has been partially supported by the project ANR-09-BLAN-0169-01  PANDA, by  the INRIA DRI  Equipe Associ\'ee PRINTEMPS and by the RAS L.R. 7/2007 project TESLA.}
}
\author{ M\'ario S. Alvim$^1$ \and Miguel E. Andr\'es$^1$ \and Konstantinos Chatzikokolakis$^1$ \and \\ Pierpaolo Degano$^2$ \and Catuscia~Palamidessi$^1$}
\institute{
$^1$ INRIA and LIX, Ecole Polytechnique, France.\\
$^2$ Dipartimento di Informatica, Universit\`a di Pisa, Italy.
}

\maketitle

\begin{abstract}
Differential privacy is a notion of  privacy that has become very popular in the database community. Roughly, the idea is that a randomized query mechanism provides sufficient privacy protection if the ratio between the probabilities that two adjacent datasets give the same answer is bound \commentMS{bounded?} by $e^\epsilon$. In the field of information flow there is a similar concern for controlling information leakage, i.e. limiting the possibility of inferring the secret information from the observables. In recent years, researchers have proposed to quantify the leakage in terms of min-entropy leakage, a concept strictly related to the Bayes risk. 
In this paper, we show how to model the query system in terms of an information-theoretic channel, and we compare the notion of differential privacy with that of min-entropy leakage. We show that  differential privacy implies a bound on the min-entropy leakage, but not vice-versa. Furthermore, we show that our bound is tight. 
Then, we consider  the utility of the randomization mechanism, which represents how close the randomized answers are to the real ones, in average.  We show that the notion of differential privacy implies a bound on utility, also tight, and we propose a method that  under certain conditions  builds an optimal randomization mechanism, i.e. a mechanism which provides the best utility while guaranteeing $\epsilon$-differential privacy. 
\end{abstract}

\section{Introduction}
\label{section:introduction}

\input{introduction}
The proofs of the results are in the appendix. 

\section{Background}
\label{section:background}

\input{background}

\section{A model of utility and privacy for statistical databases}
\label{section:model}
\input{model}

\section{Leakage}
\label{section:leakage}


\input{leakage}

\section{Utility}
\label{sec:utility}
\input{utility}

\section{Related work}
\label{section:related-work}
\input{related-work}

\section{Conclusion and future work}
\label{section:conclusion}

\input{conclusion}

\bibliographystyle{splncs}
\bibliography{short}

\newpage
\section*{Appendix}
\label{section:appendix}

\input{appendix}

\end{document}

%% file: macros.tex
\newcommand{\comment}[1]{{}}

\newcommand{\commentMS}[1]{\comment{{\bf MS: } #1}}

\newcommand{\calx}{\mathcal{X}}
\newcommand{\caly}{\mathcal{Y}}
\newcommand{\calz}{\mathcal{Z}}
\newcommand{\calk}{\mathcal{K}}
\newcommand{\calh}{\mathcal{H}}
\newcommand{\call}{\mathcal{L}}
\newcommand{\calu}{\mathcal{U}}
\newcommand{\cals}{\mathcal{S}}
\newcommand{\cala}{\mathcal{A}}
\newcommand{\calb}{\mathcal{B}}

\newcommand{\maxj}[2]{\max^{#2}\!{#1}}

\usepackage{color}
\newcommand{\resp}[1]{{\color{red}{[\emph{#1}]}}}

\def \h {h}

\newcommand{\rev}[1]{{#1}}

\newcommand{\smallsum}[1]{\textstyle{\sum_{#1}\:}}

\newenvironment{change} {} {}

%% file: introduction.tex
The area of statistical databases has been one of the first communities to consider the issues related to the protection of information. Already some decades ago, 
Dalenius  \cite{Dalenius:77:ST} proposed a famous ``ad omnia'' privacy desideratum: nothing about an individual should be learnable from the database that could not be learned without access to the database.

\paragraph{Differential privacy.}

Dalenius' property is too strong to be useful in practice: it has been shown by Dwork  \cite{Dwork:06:ICALP} that no useful database can provide it.  
In replacement, Dwork has proposed the notion of \emph{differential privacy}, which has had an extraordinary impact in the community. Intuitively, such notion is based on the idea that the presence or the absence of an individual in the database, or its particular value, should not affect in a significant way 
the probability of obtaining a certain answer for a given query \cite{Dwork:06:ICALP,Dwork:10:SODA,Dwork:11:CACM,Dwork:09:STOC}. 
\begin{change}Note that one of the important characteristics of differential privacy is that it abstracts away from the attacker's auxiliary information. The attacker might possess information about the database from external means, which could allow him to infer an individual's secret. Differential privacy ensures that no extra information can be obtained because of the individual's presence (or its particular value) in the database.\end{change}

Dwork has also studied a technique to create an $\epsilon$-differential private mechanism from an arbitrary numerical query. 
This is achieved by adding random noise to the result of the query, drawn from a Laplacian distribution with variance depending on $\epsilon$ and the query's sensitivity, i.e. the maximal difference of the query between any neighbour databases \cite{Dwork:11:CACM}. 

\paragraph{Quantitative information flow.}
The problem of preventing the leakage of secret information has been a pressing concern 
also in the area of software systems, and has motivated a very active line of research  called \emph{secure information flow}. 
In this field, similarly to the case of privacy, the goal at the beginning was ambitious:  to ensure \emph{non-interference}, 
which means complete lack of leakage. But, as for Dalenius' notion of privacy, non-interference is too strong for being obtained 
in practice, and the community has started exploring weaker notions. Some of the most popular approaches are quantitative; they do not provide a yes-or-no answer but instead try to quantify the amount of leakage using techniques from information theory. See for instance \cite{Clark:01:QAPL,Clark:05:JLC,Clarkson:09:JCS,Kopf:07:CCS,Malacaria:07:POPL,Malacaria:08:PLAS,Smith:09:FOSSACS}. 

The various approaches in the literature mainly differ on the underlying notion of entropy. Each entropy is related to the type of attacker we want to model, and to the way we measure its success (see \cite{Kopf:07:CCS} for an illuminating discussion of this relation). The most widely used is Shannon entropy~\cite{Shannon:48:Bell}, which models an adversary trying to find out the secret $x$ by asking questions of the form ``does $x$ belong to a set $S$?''. Shannon entropy is precisely the average number of questions necessary to find out the exact value of $x$ with an optimal strategy (i.e. an optimal choice of the $S$'s). 
The other most popular notion of entropy in this area is the min-entropy, proposed by R\'enyi \cite{Renyi:61:Berkeley}. The corresponding notion of attack is a \emph{single try} of the form ``is $x$ equal to value $v$?''. 
Min-entropy is precisely the logarithm of the probability of guessing the true value with the optimal strategy, which consists, of course, in selecting the $v$ with the highest probability. 
It is worth noting that the conditional min-entropy,
representing the a posteriori probability of success, is the converse of the Bayes risk~\cite{Cover:06:BOOK}.
Approaches based on min-entropy include \cite{Smith:09:FOSSACS,Braun:09:MFPS} while the Bayes risk
has been used as a measure of information leakage in \cite{Braun:08:FOSSACS,Chatzikokolakis:08:JCS}. 

\begin{change}
In this paper, we focus on the approach based on min-entropy. As it is typical in the areas of both quantitative information flow and differential privacy \cite{Kasiviswanathan:08:CORR,Ghosh:09:STC}, we model the attacker's side information as a prior distribution on the set of all databases. In our results we abstract from the side information in the sense that we prove them for all prior distributions. Note that an interesting property of min-entropy leakage is that it is maximized in the case of a uniform prior \cite{Smith:09:FOSSACS,Braun:09:MFPS}. The intuition behind this is that the leakage is maximized when the attacker's initial uncertainty is high, so there is a lot to be learned. The more information the attacker has to begin with, the less it remains to be leaked.
\end{change}

\paragraph{Goal of the paper}
The first goal of this paper is to explore the relation between differential privacy and quantitative information flow. First, we address the problem of characterizing the protection that differential privacy provides with respect to information leakage. 
Then, we consider the problem of the utility, that is the relation between the reported answer and the true answer. Clearly, a purely random result is useless, the reported answer is useful only if it provides information about the real one. It is therefore interesting to quantify the utility of the system and explore ways to improve it while preserving privacy. 
We attack this problem by considering the possible structure that the query induces on the true answers.

\paragraph{Contribution.}
The main contributions of this paper are the \rev{ following:}
\begin{itemize}
\item We propose an information-theoretic framework to reason about both information leakage and utility. 

\item We prove that $\epsilon$-differential privacy implies a bound on the information leakage. The bound is tight \begin{change}and holds for all prior distributions\end{change}. 

\item We prove that $\epsilon$-differential privacy implies a bound on the utility. We prove that, under certain conditions, the bound is tight \begin{change}and holds for all prior distributions\end{change}.

\item We identify a method that, under certain conditions, constructs the randomization mechanisms which maximizes utility while providing $\epsilon$-differential privacy. 


\end{itemize}

\paragraph{Plan of the paper.}
The next section introduces some necessary background notions. Section 3 proposes an information-theoretic view of the database query systems, and of its decomposition in terms of the query and of the randomization mechanisms. 
Section 4 shows  that differential privacy implies a bound on the min-entropy leakage, and that the bound is tight. 
Section 5 shows that differential privacy implies a bound on the utility, and that under certain conditions the bound is tight. Furthermore it shows how to construct \rev{an optimal} randomization mechanism. 
Section 6 discusses related work, and Section 7 concludes.

%% file: background.tex
This section recalls some basic notions on differential privacy and information theory.

\subsection{Differential privacy}

\begin{change}
The idea of differential privacy is that a randomized query provides sufficient privacy protection if two databases differing on a single row produce an answer with similar probabilities, i.e. probabilities whose ratio is bounded by $e^\epsilon$ for a given $\epsilon\ge 0$. More precisely:
\end{change}
\begin{definition}[\cite{Dwork:11:CACM}]
	\label{def:diff-privacy-1}
	A randomized function $\mathcal{K}$ satisfies \emph{$\epsilon$-differential privacy} if for all of data sets $D'$ and $D''$ differing on at most one row, and all $S \subseteq Range(\mathcal{K})$,
		\begin{equation}
		Pr[\mathcal{K}(D') \in S] \leq e^{\epsilon} \times Pr[\mathcal{K}(D'') \in S]		
	\end{equation}	
\end{definition}

\subsection{Information theory and interpretation in terms of attacks}

In the following, $X, Y$ denote two discrete random variables with carriers
${\cal X} = \{x_{0}{}, \ldots,$ $x_{n-1}{}\}$, ${\cal Y} = \{ y_{0}{}, \ldots, y_{m-1}{} \}$, and
probability distributions $p_{X}(\cdot)$,  $p_{Y}(\cdot)$, respectively. An information-theoretic channel is constituted of an input $X$, an output $Y$, and the matrix of conditional probabilities $p_{Y \mid X}(\cdot \mid \cdot)$, where $p_{Y \mid X}(y \mid x)$ represent the probability that $Y$ is $y$ given that $X$ is $x$. We shall omit the subscripts on the probabilities when they are clear from the context.

\paragraph{Min-entropy.}
	
	In \cite{Renyi:61:Berkeley}, R\'enyi introduced a one-parameter family of entropy measures, intended as a generalization of Shannon entropy. The R\'enyi entropy of order $\alpha$ ($\alpha > 0$,  $\alpha \neq 1$) of  a random variable $X$ is defined as $H_\alpha(X) \ =\  \frac{1}{1-\alpha}\log_2\sum_{x \,\in\,{\cal X}} p(x)^\alpha$. We are particularly interested in the limit of $H_\alpha$ as $\alpha$ approaches $\infty$. This is  called \emph{min-entropy}. It can be proven  that $H_\infty(X) \ \stackrel{\rm def}{=}\ \lim_{\alpha\rightarrow \infty}H_\alpha(X) \ =\  - \log_2\,\max_{ x\in{\cal X}}\,p(x)$.
	
	R\'enyi also defined the $\alpha$-generalization of other information-theoretic notions, like the Kullback-Leibler divergence. However, he did not define the $\alpha$-generalization of the conditional entropy, and there is no agreement on what it should be. For the case $\alpha = \infty$, we adopt here the definition proposed in \cite{Dodis:08:SIAMJoC}:
	\begin{equation}\label{eqn:SmithCondEntropyInfty}
			H_\infty(X\mid Y) \ = \  - \log_2 \smallsum{y\in {\cal Y}} p(y)\max_{x\in {\cal X}} \ p(x \mid y)
	\end{equation}
\begin{change}
We can now define the min-entropy leakage as $I _\infty = H_\infty(X) - H_\infty(X\mid Y)$. The worst-case leakage is taken by maximising over all input distributions (recall that the input distribution models the attacker's side information):  $C_\infty= \max_{p_{X}(\cdot)}I_\infty(X;Y)$. It has been proven in~\cite{Braun:09:MFPS} that $C_\infty$ is obtained at the uniform distribution, and that it is equal to the sum of the maxima of each column in the channel matrix, i.e., $C_\infty = \sum_{y \,\in\,{\cal Y}} \max_{x \,\in\,{\cal X}} p(y\mid x)$.
\end{change}

\paragraph{Interpretation in terms of attacks.}
Min-entropy can be related to a model of  adversary who is allowed to ask exactly one question of the form ``is $X = x?$'' (one-try attack). More precisely, $H_\infty(X)$ represents the (logarithm of the inverse of the) probability of success for this kind of attacks with the best strategy, which consists, of course, in choosing the $x$ with the maximum probability. 

The conditional min-entropy $H_\infty(X\mid Y)$ represents (the logarithm of the inverse of) the probability that the same kind of adversary  succeeds in guessing the value of $X$ \emph{a posteriori}, i.e.  after observing the  result of $Y$. The complement of this probability is also known as \emph{probability of error} or \emph{Bayes risk}. Since in general $X$ and $Y$ are correlated, observing $Y$ increases the probability of success. Indeed we can prove formally that $H_\infty(X\mid Y) \leq H_\infty(X)$, with equality if and only if $X$ and $Y$ are independent. The min-entropy leakage $I_\infty(X;Y) = H_{\infty}(X) - H_{\infty}(X|Y)$ corresponds to the \emph{ratio} between the probabilities of success a priori and a posteriori, which is a natural notion of leakage. Note that it is always the case that $I_\infty(X;Y)\geq 0$, which seems desirable for a good notion of leakage.

%% file: model.tex
In this section we present a model of statistical queries on databases, where
noise is carefully added to protect privacy and, in general, the 
reported answer to a query does not need to correspond to the real one. 
In this model, the notion of information leakage can be used to measure the amount \rev{of} information that an attacker can learn about
the database by posting queries and analysing their (reported) answers.  Moreover, the model allows us to quantify the utility of the
query, that is, how much information about the real answer can be obtained from
the reported one. This model will serve as the basis for exploring the relation
between differential privacy and information flow.

\begin{change}
We fix a finite set $Ind=\lbrace 1,2,\ldots, u\rbrace$ of $u$ individuals participating in the database. In addition,  we fix a  finite set $Val=\lbrace ${\sl{v}}$_1,${\sl v}$_2,\ldots,${\sl v}$_v\rbrace$, representing the set of ($v$ different) possible values for the \emph{sensitive attribute} of each individual (e.g. disease-name in a medical database)\footnote{In case there are several sensitive attributes in the database (e.g. \rev{skin color and presence of a certain medical condition}),  we can think of the elements of $Val$ as tuples.}. Note that the absence of an
individual from the database, if allowed, can be modeled with a special
value in $Val$. As usual in the area of differential privacy \cite{Nissim:07:STOC},
we model a database as a $u$-tuple $D= \{d_0, \ldots, d_{u-1} \}$
where each $d_i \in {\it Val}$ is the value of the corresponding individual.
The set of all databases is $\calx = {\it Val}^u$. Two databases $D,D'$ are
\emph{adjacent}, written $D\sim D'$ iff they differ for the value of exactly one
individual.
\end{change}

\begin{wrapfigure}{r}{0.40\textwidth}
\vspace{-0.5cm}
	\centering
	\includegraphics[width=0.32\textwidth]{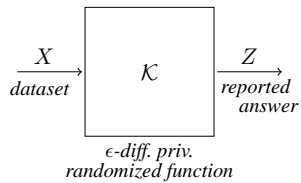}%
	\caption{\!Randomized function $\mathcal{K}$ as a channel}%
	\label{fig:mechanism-k}%
	\vspace{-0.8cm}
\end{wrapfigure}
Let $\calk$ be a randomized function from $\calx$ to $\calz$, where $\calz=Range(\calk)$ (see Figure~\ref{fig:mechanism-k}). This function can be
modeled by a channel with input and output alphabets $\calx,\calz$
respectively. This channel can be specified as usual by a matrix of conditional probabilities $p_{Z|X}(\cdot|\cdot)$. We also denote
by $X,Z$ the random variables modeling the input and output of the channel. The definition of differential privacy can be directly expressed as a property
of the channel: it satisfies $\epsilon$-differential privacy iff
\[ p(z|x)\le e^\epsilon p(z|x') \textrm{ for all } z\in \calz, x,x'\in \calx
\textrm{ with } x\sim x'
\]

Intuitively, the \emph{correlation} between $X$ and $Z$ measures how much
information about the complete database the attacker can obtain by observing the
reported answer. We will refer to this correlation as the \emph{leakage}
of the channel, denoted by $\call(X,Z)$. In Section~\ref{section:leakage}
we discuss how this leakage can be quantified, using notions from information
theory, and we study the behavior of the leakage for differentially private
queries.

We then introduce a random variable $Y$ modeling the true answer to the query
$f$, ranging over $\caly = Range(f)$. The correlation between $Y$ and $Z$
measures how much we can learn about the real answer from the reported one. We
will refer to this correlation as the \emph{utility} of the channel, denoted by
$\calu(Y,Z)$. In Section~\ref{sec:utility} we discuss in detail how utility
can be quantified, and we investigate how to construct a randomization mechanism, i.e. a way of adding noise to the query outputs,
so that utility is maximized while preserving differential privacy.

In practice, the randomization mechanism is
often \emph{oblivious}, meaning that the reported answer $Z$ only depends on
the real answer $Y$ and not on the database $X$. In this case, the randomized function $\mathcal{K}$, seen as 
a channel, can be decomposed into two parts: a channel modeling
the query $f$, and a channel modeling the oblivious randomization mechanism $\calh$. The
definition of utility in this case is simplified as it only depends on properties of the
sub-channel correspondent to $\mathcal{H}$. The leakage relating $X$ and $Y$ and the utility relating $Y$ and $Z$ for a decomposed randomized function
are shown in Figure~\ref{fig:utility-privacy}.
\begin{figure}[tb]
	\centering
	\includegraphics[width=0.7\textwidth]{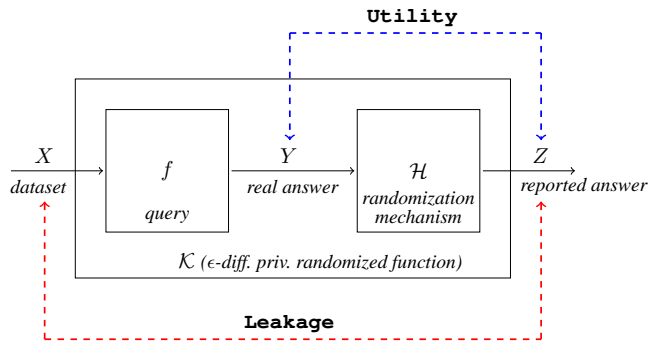}%
	\caption{Leakage and utility for oblivious mechanisms}%
	\label{fig:utility-privacy}%
\end{figure}

\paragraph{Leakage about an individual.}
\begin{change}
As already discussed, $\call(X,Z)$ can be used to quantify the amount of information about the
\emph{whole database} that is leaked to the attacker.
However, protecting the database as a whole
is not the main goal of differential privacy. Indeed, some information is allowed by design to be revealed,
otherwise the query would not be useful. Instead,
differential privacy aims at protecting the value of \rev{each} individual.
Although $\call(X,Z)$ is a good measure of the overall privacy of the system,
we might be interested in measuring how much information about a \emph{single individual} is leaked.

To quantify this leakage, we assume that the values of all other individuals are
already known, thus the only remaining information concerns the individual of
interest. Then we define smaller channels, where only the information of a
specific individual varies. Let $D^- \in {\it Val}^{u-1}$ be a $(u-1)$-tuple
with the values of all individuals except the one of interest. We
create a channel $\mathcal{K}_{D^-}$ whose input alphabet is the set of all
databases in which the $u-1$ other individuals have the same values as in $D^-$.
Intuitively, the information leakage of this channel measures how much
information about one particular individual the attacker can learn if the values
of all others are known to be $D^-$. This leakage is studied in
Section~\ref{sec:individual-leakage}.
\end{change}

%% file: leakage.tex


As discussed in the previous section, the correlation $\call(X,Z)$ between $X$
and $Z$ measures the information that the attacker can learn about the database by observing the reported answers.
In this section, we consider min-entropy leakage as a measure of this information,
that is $\call(X,Z) = I_\infty(X;Z)$. We then investigate bounds on information
leakage imposed by differential privacy. These bounds hold for any side information of the attacker, modelled
as a prior distribution on the inputs of the channel.

Our first result shows that the min-entropy leakage of a randomized function $\cal K$ is
bounded by a quantity depending on $\epsilon$, the numbers $u,v$ of individuals and
values respectively. We assume that $v\geq 2$.

\begin{theorem}\label{theo:bound-leakage}
If $\mathcal{K}$ provides \emph{$\epsilon$-differential privacy} then \begin{change}for all input distributions,\end{change}
the min-entropy leakage associated to $\mathcal{K}$ is bounded from above as follows:
\[I_\infty(X;Z)  \leq u\, \log_2\frac{v\, e^\epsilon}{(v-1 +  e^\epsilon)}\]
\end{theorem}

\rev{
}

\begin{wrapfigure}{r}{0.46\textwidth}
\vspace{-0.73cm}
		\centering
		\includegraphics[width=0.36\textwidth]{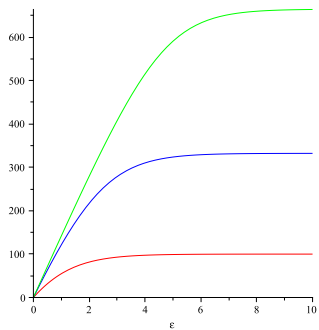}%
		\caption{Graphs of $B(u,v,\epsilon)$ for $u\!\!=\!\!100$ and $v\!\!=\!\!2$ (lowest line), $v\!\!=\!\!10$ (intermediate line), and $v\!\!=\!\!100$ (highest line), respectively.}%
		\label{fig:plots}%
		\vspace{-1.25cm}
\end{wrapfigure}
 \bigskip
Note that this bound $B(u,v,\epsilon) = u\, \log_2\frac{v\, e^\epsilon}{(v-1 +  e^\epsilon)}$  is a continuous function in $\epsilon$, has value $0$ when $\epsilon=0$, 
and converges to $u\, \log_2 v$ as $\epsilon$ approaches infinity. 
Figure~\ref{fig:plots} shows the growth of  $B(u,v,\epsilon)$ along with $\epsilon$, for various fixed values of $u$ and $v$.

\medskip
The following result shows that the bound $B(u,v,\epsilon)$ is \emph{tight}.

\begin{proposition}\label{prop:tight}
For every $u$, $v$, and $\epsilon$ there exists a  randomized function $\mathcal K$
which provides $\epsilon$-differential privacy and whose min-entropy leakage is $I_\infty(X;Z)=B(u,v,\epsilon)$ for the uniform input distribution.
\end{proposition}

\begin{example}\label{exa:eyes}
Assume that we are interested in the \rev{eye} color of a certain population 
 $\mathit{Ind}=\{{\it Alice}, {\it Bob}\}$. Let $\mathit{Val}=\{\mathtt{a},\mathtt{b},\mathtt{c}\}$ where $\mathtt{a}$ stands for $\mathit{absent}$ (i.e. the \emph{null} value), $\mathtt{b}$ stands for $\mathit{blue}$, and $\mathtt{c}$ stands for $\mathit{coal}$ (black). 
 We can represent each dataset with a tuple $d_1 d_0$, 
 where $d_0\in \mathit{Val}$ represents the \rev{eye} color of  ${\it Alice}$ (cases $d_0=b$ and $d_0=c$), 
or that ${\it Alice}$ is not in the dataset (case $d_0=a$). The value $d_1$ provides the same kind of information for $\mathit{Bob}$. Note that $v=3$.
Fig~\ref{fig:eyes.a} represents the set $\cal X$ of all possible datasets and its adjacency relation. 
\rev{We now construct the matrix with input $\cal X$ which provides  $\epsilon$-differential privacy and has the highest min-entropy leakage.
From the proof of Proposition~\ref{prop:tight}, we know that each element of the matrix  is of the form $\frac{a}{e^{\epsilon \,  d}}$, where $a$ is the highest value in the matrix, i.e. $a= \frac{v\, e^\epsilon}{(v-1 +  e^\epsilon)} =
 \frac{3\, e^\epsilon}{(2 +  e^\epsilon)}$, and $d$ is the graph-distance (in Fig~\ref{fig:eyes.a}) between (the dataset of) the row which contains such element and (the dataset of) the row with the highest value in the same column. Fig~\ref{fig:eyes.b} illustrates this matrix, where, for the sake of readability, each value $\frac{a}{e^{\epsilon \, d}}$ is represented simply by $d$.}
 \end{example}

\vspace{-0.25cm}
  \begin{figure}[!htb]%
  \centering
		 \subfigure[
		 {The datasets and their adjacency relation}
		 ]{
                    \centering
                    \includegraphics[width=0.3\linewidth]{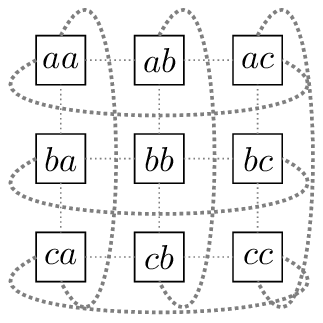}%
                    \label{fig:eyes.a}%
                } \hspace{2cm}
		 \subfigure[
		 {The representation of the matrix}]{
                    \centering
                    \includegraphics[width=0.3\linewidth]{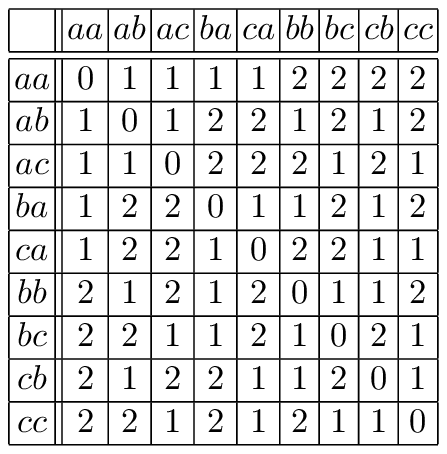}%
                    \label{fig:eyes.b}%
                }
		\caption{Universe and highest min-entropy leakage matrix giving $\epsilon$-differential privacy for Example~\ref{exa:eyes}.}%
\end{figure}

Note that the bound $B(u,v,\epsilon)$ is guaranteed to be reached with the uniform input distribution. We know from the literature \cite{Braun:09:MFPS,Smith:09:FOSSACS} that the $I_\infty$ of a given matrix has its maximum in correspondence of the uniform input distribution, although it may not be the only case. 

The construction of the matrix for Proposition~\ref{prop:tight} gives a square matrix of dimension $v^u \times v^u$. Often, however, the range of $\mathcal K$ is fixed, 
as it is usually related to the possible answers to the query $f$. Hence it is natural to consider the scenario in which 
we are  given a number $r<v^u$, and want to consider only those  $\mathcal K$'s whose range has  cardinality  at most $r$. In this restricted setting, we could find a better bound than the one given by Theorem~\ref{theo:bound-leakage},
as the following proposition shows.



\begin{proposition}\label{prop:new-bound}
Let $\mathcal K$ be a randomized function and let $r = |\mathit{Range}({\mathcal{K})}|$.  
If $\mathcal{K}$ provides \emph{$\epsilon$-differential privacy} then \begin{change}for all input distributions\end{change},
the min-entropy leakage associated to $\mathcal{K}$ is bounded from above as follows:
\[I_\infty(X;Z) \, \leq \,\log_2\frac{r\,(e^{\epsilon})^u}{(v-1+e^\epsilon)^\ell-(e^{\epsilon})^\ell+(e^{\epsilon})^u}\]
where $\ell=\lfloor \log_v r\rfloor$.
\end{proposition}

Note that this bound can be much smaller than the one provided by Theorem~\ref{theo:bound-leakage}. For instance, if $r=v$ 
 this bound becomes:
 \begin{equation*} \log_2 \frac{v\,(e^{\epsilon})^u}{v-1+(e^{\epsilon})^u} \end{equation*}
which  for large  values of $u$ is
much smaller than $B(u,v,\epsilon)$. \rev{In particular, for $v=2$ and $u$ approaching infinity, this bound approaches $1$, while $B(u,v,\epsilon)$ approaches infinity. }

 Let us clarify that there is no contradiction with the fact that the bound $B(u,v,\epsilon)$ is strict: indeed it is strict when we are free to choose the range, but here we fix the dimension of the range.

Finally, note that the above bounds do not hold in the opposite direction. Since min-entropy averages over all observations, low probability observations affect it only slightly. Thus, by introducing an observation with a negligible probability for one user, and zero probability for some other user, we could have a channel with arbitrarily low min-entropy leakage but which does not satisfy differential privacy for any $\epsilon$.
 

\subsection{Measuring the leakage about an individual}
\label{sec:individual-leakage}

As discussed in Section~\ref{section:model}, the main goal of differential
privacy is not to protect information about the complete database, but about
each individual. To capture the leakage about a certain individual, we
start from a tuple $D^- \in {\it Val}^{u-1}$ containing the given (and known) values of all
other $u-1$ individuals. Then we create a channel whose input $X_{D^-}$ ranges
over all databases where the values of the other   individuals are exactly those of $D^-$
and only the value of the selected  individual varies. Intuitively,
$I_{\infty}(X_{D^-};Z)$ measures the leakage about the individual's value where
all other values are known to be as in $D^-$. As all these databases are adjacent,
differential privacy provides a stronger bound for this leakage.

\begin{theorem}
	\label{theo:Renyi}
	If $\mathcal{K}$ provides \emph{$\epsilon$-differential privacy} then 
	for all $D^- \in {\it Val}^{u-1}$ and \begin{change}for all input distributions\end{change},
	the min-entropy leakage about an individual is bounded from above as follows:

	\[ I_{\infty}(X_{D^-};Z) \leq \log_2 e^\epsilon \]
\end{theorem}

Note that this bound is stronger than the one of
Theorem~\ref{theo:bound-leakage}. In particular, it depends only on $\epsilon$
and not on $u,v$.

%% file: utility.tex

	As discussed in Section~\ref{section:model}, the utility of a randomized function $\mathcal{K}$ is the correlation between the real answers $Y$ for a query and the reported answers $Z$. In this section we analyze the utility $\calu(Y,Z)$ using the classic notion of \emph{utility functions} (see for instance \cite{Bernardo:94:BOOK}). 

	
	\commentMS{Discuss if the user can repeat the query. Dwork says that in principle repeated queries should give the same answer, which validates this reasoning. However, we don't know if it's feasible and it might not be a reasonable assumption.}
	
For our analysis we assume an oblivious randomization mechanism. As
discussed in Section~\ref{section:model}, in this case the system can be
decomposed into two channels, and the utility becomes a property of the channel associated to the randomization mechanism $\calh$ which maps the real
answer $y \in \mathcal{Y}$ into a reported answer $z \in \mathcal{Z}$ according
to given probability distributions $p_{Z|Y}(\cdot|\cdot)$. However, the user does not necessarily take $z$ as her guess for the real answer, since she can use some Bayesian post-processing to maximize the probability of success, i.e. a right guess. Thus for each reported answer $z$ the user can remap her guess to a value $y' \in \mathcal{Y}$ according to a remapping function $\rho(z) : \mathcal{Z} \rightarrow \mathcal{Y}$, that maximizes her expected gain. For each pair $(y,y')$, with $y\in \caly, y'=\rho(y)$, there is an associated value given by a gain (or utility) function $g(y,y')$ that represents a score of how useful it is for the user to guess the value $y'$ as the answer when the real answer is $y$.

It is natural to define the global utility of the mechanism $\cal H$   as the expected gain:

	\begin{equation}
		\mathcal{U}(Y,Z) = \displaystyle \sum_{y} p(y) \sum_{y'} p(y'|y) g(y,y') 
		\label{eq:gain}
	\end{equation}
	
\noindent{where $p(y)$ is the prior probability of real answer $y$, and $p(y'|y)$ is the probability of user guessing $y'$ when the real answer is $y$.}
	
We can derive the following characterization of the utility.  We use  $\delta_x$ to represent  the probability distribution which has value $1$ on $x$ and $0$ elsewhere.
\begin{align*}
	\mathcal{U}(Y,Z)  & =    \displaystyle \sum_{y} p(y) \sum_{y'} p(y'|y) g(y,y') & \text{(by \eqref{eq:gain})} \\
			     							    & =    \displaystyle \sum_{y} p(y) \sum_{y'} \left( \sum_{z} p(z|y)p(y'|z) \right) g(y,y') \\
			     							    & =    \displaystyle \sum_{y} p(y) \sum_{y'} \left( \sum_{z} p(z|y) \delta_{\rho(z)}(y') \right) g(y,y') & \text{(by remap $y'= \rho(z)$)} \\
			     							    & =    \displaystyle \sum_{y} p(y) \sum_{z} p(z|y) \sum_{y'} \delta_{\rho(z)}(y')g(y,y') \\
			     							    & =    \displaystyle \sum_{y,z} p(y,z) \sum_{y'} \delta_{\rho(z)}(y')g(y,y')  \\
			     							    & =    \displaystyle \sum_{y,z} p(y,z) g(y,\rho(z)) 
\end{align*}
 
A very common utility function is the \emph{binary gain function}, which is  defined as $g_{\rm bin}(y,y') = 1$ if $y=y'$ and $g_{\rm bin}(y,y') =0$ if $y \neq y'$. 
The rationale behind this function is that, when the answer domain does not have a notion of distance, then the wrong answers are all equally bad. Hence the gain is  total   when we guess the exact answer, and is $0$ for all other guesses.
Note that if the answer domain  is equipped with a notion of distance, 
then the gain function could take into account the proximity of the reported answer to the real one, the idea being that a close answer, even if  wrong, is better than a distant one.

In this paper we do not assume a notion of distance, and we will focus on the binary case. 
The use of binary utility functions in the context of
differential privacy was also investigated  
in \cite{Ghosh:09:STC}\footnote{\rev{Instead of gain functions, \cite{Ghosh:09:STC} equivalently 
uses the dual notion of \emph{loss functions}}.}.

By substituting $g$ with $g_{\rm bin}$ in the above formula   we obtain:
	
	\begin{equation}
		\mathcal{U}(Y,Z) = \displaystyle \sum_{y,z} p(y,z) \delta_{y}(\rho(z))
		\label{eq:utility-function}
	\end{equation}
\noindent{which tells us that the expected utility is the greatest when $\rho(z) = y$ is chosen to maximize $p(y,z)$. Assuming that the user chooses such a maximizing remapping, we have:}
	
	\begin{equation}
		\begin{array}{lcll}
			\displaystyle \mathcal{U}(Y,Z) & =    & \displaystyle \sum_{z} \max_{y} p(y,z) & \mbox{} \\
		\end{array}
		\label{eq:utility-renyi}
	\end{equation}
This corresponds to the converse of the Bayes risk, and it is closely related to the conditional min-entropy and to the min-entropy leakage: 
\[H_\infty(Y|Z) = -\log_2 \mathcal{U}(Y,Z) \qquad\qquad I_\infty(Y;Z) = H_\infty(X) + \log_2 \mathcal{U}(Y,Z)\]

\subsection{A bound on the utility}

In this section we show that the fact that $\cal K$ provides $\epsilon$-differential privacy induces a bound on the utility. 
We start by extending the adjacency relation $\sim$ from the datasets $\calx$ to the answers $\cal Y$. 
Intuitively, the function $f$ associated to the query determines a partition on
the set of all databases ($\cal X$, i.e. ${\it Val}^u$), and we say that two classes are adjacent if 
they contain an adjacent pair. More formally: 

\begin{definition}
Given $y,y'\in{\cal Y}$, with $y\neq y'$, we say that $y$ and $y'$ are adjacent
(notation $y\sim y'$), iff there exist $D,D'\in {\it Val}^u$ with $D\sim D'$ such that 
$ y = f(D)$ and  $y'=f(D')$.
\end{definition}
Since $\sim$ is symmetric on databases, it is also symmetric on $\caly$, therefore also $({\cal Y},\sim)$ forms an undirected graph.
\begin{definition}
	\label{def:distance-border}
	The \emph{distance} $\mathit{dist}$ between two elements $y,y'\in{\cal Y}$ is the length of the minimum path from $y$ to $y'$.
	For a given natural number $d$, we define ${\it Border}_{d}(y)$ as the set of elements at distance $d$ from $y$:
	\[{\it Border}_{d}(y)=\{y'\mid {\it dist}(y,y')=d\}\]
\end{definition}
We recall that a graph automorphism is a permutation of its vertices that preserves its edges. If $\sigma$ is a permutation of $S$ then an orbit of $\sigma$ is a set of the form $\{\sigma^i(s)\,|\,i \in\mathbb{N}\}$ where $s \in S$. A permutation has a single orbit iff $\{ \sigma^i(s) | i \in \mathbb{N} \} = S$ for all $s\in S$. 

\rev{The next theorem provides a bound on the utility in the case in which $({\cal Y},\sim)$ admits a graph automorphism with a single orbit. Note that this condition implies that the  graph has a very regular structure; in particular, all nodes must have the same number of incident edges. Examples of such graphs are rings and cliques (but they are not the only cases). }

\begin{theorem}\label{theorem:utility-bound}
Let $\cal H$ be a randomization mechanism for the randomized function $\cal K$ and the query $f$, and assume that $\cal K$ provides $\epsilon$-differential privacy. 
Assume that $({\cal Y},\sim)$ admits a graph automorphism with a single orbit. 
Furthermore, assume that there exists a natural number $c$ and an element $y\in{\cal Y}$ such that, \rev{for every natural number $d>0$}, either $|{\it Border}_{d}(y)| = 0$ or $|{\it Border}_{d}(y)| \geq c$. Then
\[
\calu(X,Y) \leq \frac{(e^\epsilon)^n(1-e^\epsilon)}{(e^\epsilon)^n(1-e^\epsilon) + c\,(1-(e^\epsilon)^n)}
\]
where $n$ is the maximum distance from $y$ in $\cal Y$.
\end{theorem}

\rev{
The bound provided by the above theorem is strict in the sense that for every $\epsilon$  and ${\cal Y}$ there exist an adjacency relation  $\sim$ for which we can construct a randomization mechanism $\cal H$  that  provides $\epsilon$-differential privacy and whose utility achieves the bound of Theorem~\ref{theorem:utility-bound}. This randomization mechanism is therefore optimal, in the sense that it provides the maximum possible utility 
for the given $\epsilon$.
Intuitively, the condition on $\sim$ is that    $|{\it Border}_{d}(y)|$ must be exactly $c$ or $0$  for every $d>0$. 
In the next section we will define formally such an optimal randomization mechanism, and give examples of queries that determine a relation $\sim$ satisfying the condition. 
\subsection{Constructing an optimal randomization mechanism}
Assume $f:\calx\rightarrow \caly$, and consider the graph structure $(\caly,\sim)$ determined by $f$.  Let $n$ be the maximum distance between two nodes in the graph and let $c$ be an integer. 
We construct the matrix $M$ of conditional probabilities associated to $\cal H$ as follows. For every column $z \in \calz$ and every row $y \in \caly$, define:
\begin{equation}
\label{eqn:optimal-matrix}
p_{Z|Y}(z|y)  =  {\alpha}/{(e^\epsilon)^d}   \mbox{ where } d = \mathit{dist}(y,z)\mbox{ and }
\alpha  = \frac{\scriptstyle{(e^\epsilon)^n(1-e^\epsilon)}}{\scriptstyle{(e^\epsilon)^n(1-e^\epsilon) + 
		c\,(1-(e^\epsilon)^n)}}
\end{equation}
}

The following theorem guarantees that the randomization mechanism ${\cal H}$ defined above is well defined and optimal, under certain conditions.

\begin{theorem}
\label{theorem:construction-h}
Let $f:\calx\rightarrow \caly$ be a query and  let $\epsilon\geq 0$. Assume that $({\cal Y},\sim)$  admits a graph automorphism with a single orbit, and that there exists $c$ such that, for every $y\in \caly$ and \rev{every natural number $d>0$}, either $|{\it Border}_{d}(y)| = 0$ or $|{\it Border}_{d}(y)| = c$. Then, for such $c$,  the definition  in (\ref{eqn:optimal-matrix}) determines a legal channel matrix for $\cal H$, i.e.,  for each $y \in \caly$, $p_{Z|Y}(\cdot|y)$ is a probability distribution. Furthermore,   the composition  $\cal K$ of $f$ and $\cal H$ provides $\epsilon$-differential privacy. Finally, $\cal H$ is optimal in the sense that it maximizes utility when the  distribution of $Y$ is uniform. 
\end{theorem}

The conditions for the construction of the optimal matrix are strong, but there are some interesting cases in which they are satisfied. 
Depending on the degree of connectivity $c$, \rev{we can have several different cases whose extremes are:}
\begin{itemize}
\item $({\cal Y},\sim)$ is a \emph{ring}, i.e. every element has exactly two adjacent elements. This is similar to the case of the counting queries considered in \cite{Ghosh:09:STC}, with the difference that our ``counting'' is in arithmetic modulo $|{\cal Y}|$. 
\item $({\cal Y},\sim)$ is a \emph{clique}, i.e. every element has exactly $|{\cal Y}|-1$ adjacent elements. 
\end{itemize}

\begin{remark}\label{rem:relax1}
Note that when we have a ring with an even number of nodes the conditions of Theorem~\ref{theorem:construction-h}
are almost met, except that  $|{\it Border}_{d}(y)| = 2$ for $d<n$, and $|{\it Border}_{d}(y)| = 1$ for $d=n$, where $n$ is the maximum distance between two nodes in $\caly$.
In this case, and if $(e^\epsilon)^2\geq 2$, we  can still construct a legal matrix by doubling the value of such elements. Namely, by defining
\[
p_{Z|Y}(z|y) = 2 \frac{\alpha}{(e^\epsilon)^n}  \qquad \mbox{if }  \mathit{dist}(y,z)=n
\]
For all the other elements the definition remains as in (\ref{eqn:optimal-matrix}).
\end{remark}

\begin{remark}\label{rem:relax2}
Note that our method can be applied also when the conditions of 
Theorem~\ref{theorem:construction-h} are not met: We can always  
add ``artificial'' adjacencies to the graph structure so to meet those conditions. 
Namely, for computing the distance in (\ref{eqn:optimal-matrix}) we use, 
instead of $(\caly,\sim)$, a structure $(\caly,\sim')$ which satisfies the conditions of 
Theorem~\ref{theorem:construction-h}, and such that $\sim\, \subseteq \, \sim'$. 
Naturally, the matrix constructed in this way provides $\epsilon$-differential privacy, but in general is not optimal. 
Of course, the smaller $\sim'$ is, the \rev{higher} is the utility. 
\end{remark}

The  matrices generated by our algorithm above can be very different, depending on the value  of $c$. 
The next two examples illustrate queries that give rise to the clique and to the ring structures, and show the corresponding matrices. 

\begin{example} Consider a database with electoral information where rows corresponds to voters. Let us assume, for simplicity, that each row contains only three fields:

\begin{itemize}
\item ID: a unique (anonymized) identifier assigned to each voter;
\item CITY: the name of the city where the user voted;
\item CANDIDATE: the name of the candidate the user voted for.
\end{itemize}

Consider the query \emph{``What is the city with the greatest number of votes for a given candidate?''}. For this query the binary function is a natural choice for the gain function: only the right city gives some gain, and any wrong answer is just as bad as any other. 

It is easy to see that every two answers are neighbors, i.e. \emph{the graph structure of the answers is a clique}. 

Consider the case where CITY=\{A,B,C,D,E,F\} and assume for simplicity that there is a unique answer for the query, i.e., there are no two cities with exactly the same number of individuals voting for a given candidate. Table~\ref{tab:city-cand-geo} shows two alternative mechanisms providing $\epsilon$-differential privacy (with $\epsilon = \log 2$). The first one, $M_{1}$, is based on the truncated geometric mechanism method used in \cite{Ghosh:09:STC} for counting queries (here extended to the case where every two answers are neighbors). The second mechanism, $M_{2}$, is the one we propose in this paper.

\begin{table}[tb]
\centering	
	\subtable[$M_{1}$: truncated geometric mechanism]{	
		$
			\begin{array}{|c||c|c|c|c|c|c|}
				\hline
				\mbox{In/Out} & A    & B    & C    & D    & E    & F    \\ \hline \hline
				A & 0.535 & 0.060 & 0.052 & 0.046 & 0.040 & 0.267 \\ \hline
				B & 0.465 & 0.069 & 0.060 & 0.053 & 0.046 & 0.307 \\ \hline
				C & 0.405 & 0.060 & 0.069 & 0.060 & 0.053 & 0.353 \\ \hline
				D & 0.353 & 0.053 & 0.060 & 0.069 & 0.060 & 0.405 \\ \hline
				E & 0.307 & 0.046 & 0.053 & 0.060 & 0.069 & 0.465 \\ \hline
				F & 0.267 & 0.040 & 0.046 & 0.052 & 0.060 & 0.535 \\ \hline
			\end{array}
		$
		\label{tab:city-cand-geo-a}
	}
	\subtable[$M_{2}$: our mechanism]{		
		$
			\begin{array}{|c||c|c|c|c|c|c|}
			 \hline
			 \mbox{In/Out} & A    & B    & C    & D    & E    & F    \\ \hline \hline
			 A             & 2/7  & 1/7 & 1/7 & 1/7 & 1/7 & 1/7 \\ \hline
			 B             & 1/7  & 2/7 & 1/7  & 1/7 & 1/7 & 1/7 \\ \hline
			 C             & 1/7 & 1/7 & 2/7  & 1/7  & 1/7 & 1/7 \\ \hline
			 D             & 1/7  & 1/7 & 1/7  & 2/7  & 1/7  & 1/7  \\ \hline
			 E             & 1/7 & 1/7 & 1/7 & 1/7  & 2/7 & 1/7 \\ \hline
			 F             & 1/7 & 1/7 & 1/7 & 1/7 & 1/7 & 2/7  \\ \hline
			\end{array}
		$
		\label{tab:city-cand-geo-b}
	}
	\caption{Mechanisms for the city with higher number of votes for a given candidate}
	\vspace{-0.70cm}
	\label{tab:city-cand-geo}
\end{table}

Taking the input distribution, i.e. the distribution on $Y$, as the uniform distribution, it is easy to see that ${\cal U}(M_1) = 0.2243 < 0.2857 = {\cal U}(M_2)$. Even for non-uniform distributions, our mechanism still provides better utility. For instance, for $p(A) = p(F) = 1/10$ and $p(B) = p(C) = p(D) = P(E) = 1/5$, we have ${\cal U}(M_1) = 0.2412 < 0.2857 = {\cal U}(M_2)$. This is not too surprising: the Laplacian method and the geometric mechanism work  very well when the domain of answers is provided with a metric and the utility function takes into account the proximity of the reported answer to the real one. It also works well when $({\cal Y}, \sim)$ has low connectivity, in particular in the cases of a ring and of a line. But in this example, we are not in these cases, because we are considering \emph{binary gain functions} and \emph{high connectivity}.

\end{example}

\begin{example}

Consider the same database as the previous example, but now assume a counting query of the form \emph{``What is the number of votes for candidate $\mathit{cand}$?''}. 
It is easy to see that each answer has at most two neighbors. More precisely,\emph{ the graph structure on  the answers is a line}. 
For illustration purposes, let us assume that only $5$ individuals have participated in the election. Table~\ref{tab:count-geo} shows two alternative mechanisms providing $\epsilon$-differential privacy ($\epsilon = \log 2$): (a) the truncated geometric mechanism  $M_{1}$ proposed in \cite{Ghosh:09:STC} and (b) the mechanism $M_2$ that we propose, where $c=2$ and $n=3$. Note that in order to apply our method we have first to apply Remark~\ref{rem:relax2} to transform the line into a ring, and then Remark~\ref{rem:relax1} to 
handle the case of the elements at maximal distance from the diagonal.


Le us consider the uniform prior distribution. We see that the utility of $M_1$ is higher than the utility of $M_2$, in fact the first is $4/9$ and the second is $4/11$. 
This does not contradict our theorem, because our matrix is guaranteed to be optimal only in the case of a ring structure, not a line as we have in this example.
If the structure were a ring, i.e. if the last row were adjacent to the first one, then $M_1$ would not provide $\epsilon$-differential privacy. 
In case of a line as in this example,   the truncated geometric mechanism  has been proved optimal \cite{Ghosh:09:STC}. 
\begin{table}[tb]
	\centering
	\subtable[$M_{1}$: truncated $\frac{1}{2}$-geom. mechanism]{
		$
			\begin{array}{|c||c|c|c|c|c|c|}
				\hline
				\mbox{In/Out} & 0    & 1    & 2    & 3    & 4    & 5    \\ \hline \hline
				0                   & 2/3  & 1/6  & 1/12 & 1/24 & 1/48 & 1/48 \\ \hline
				1                   & 1/3  & 1/3  & 1/6  & 1/12 & 1/24 & 1/24 \\ \hline
				2                   & 1/6  & 1/6  & 1/3  & 1/6  & 1/12 & 1/12 \\ \hline
				3                   & 1/12 & 1/12 & 1/6  & 1/3  & 1/6  & 1/6  \\ \hline
				4                   & 1/24 & 1/24 & 1/12 & 1/6  & 1/3  & 1/3  \\ \hline
				5                   & 1/48 & 1/48 & 1/24 & 1/12 & 1/6  & 2/3  \\ \hline
			\end{array}
		$
		\label{tab:count-geo-a}
	}
	\subtable[$M_{2}$: our mechanism]{
		$
			\begin{array}{|c||c|c|c|c|c|c|}
			 	\hline
				 \mbox{In/Out} & 0    & 1    & 2    & 3    & 4    & 5    \\ \hline \hline
				 0          & 4/11  & 2/11 & 1/11 & 1/11 & 1/11 & 2/11 \\ \hline
				 1          & 2/11  & 4/11 & 2/11  & 1/11 & 1/11 & 1/11 \\ \hline
				 2          & 1/11  & 2/11 & 4/11  & 2/11  & 1/11 & 1/11 \\ \hline
				 3          & 1/11  & 1/11 & 2/11  & 4/11  & 2/11  & 1/11  \\ \hline
				 4          & 1/11  & 1/11 & 1/11 & 2/11  & 4/11  & 2/11  \\ \hline
				 5          & 2/11  & 1/11 & 1/11 & 1/11 & 2/11 & 4/11  \\ \hline
			\end{array}
		$
		\label{tab:count-geo-b}
	}
	\caption{Mechanisms for the counting query ($5$ voters)}
	\vspace{-0.70cm}
	\label{tab:count-geo}	
\end{table}
\end{example}

%% file: related-work.tex
As far as we know, the first work to investigate the relation between
differential privacy and information-theoretic leakage \emph{for an
individual} was \cite{Alvim:10:TechRep}. In this work, a channel is relative to a given
database $x$, and the channel inputs  are all possible databases
adjacent to $x$. Two bounds on leakage were presented, one for the Shannon entropy, and one for the min-entropy. 
The latter corresponds to Theorem~\ref{theo:Renyi} in this paper (note that  \cite{Alvim:10:TechRep} is an unpublished report).

Barthe and K\"opf \cite{Barthe:11:CSF} were the first to investigates the (more
challenging) connection between differential privacy and the 
min-entropy leakage \emph{for the entire universe of possible databases}.
They consider only  the hiding of the \emph{participation} of
individuals in a database, which corresponds to the  case of $v=2$ in our setting. 
They consider the ``end-to-end differentially private mechanisms'',
which correspond to what we call $\cal K$ in
our paper, and propose, like we do, to interpret them as information-theoretic
channels. 
They provide a bound for the leakage, but  point out that it  is not tight in general, and show that there cannot be a
domain-independent bound, by proving that for any number of individual
$u$ the optimal bound must be at least a certain expression $f(u,\epsilon)$.
Finally, they show that the question of providing optimal upper bounds
for the leakage of $\cal K$  in
terms of rational functions of $\epsilon$ is decidable, and leave the actual
function as an open question. In our work we used rather different
techniques and found (independently) the same function $f(u,\epsilon)$ 
(the bound $B(u,v,\epsilon)$ in Theorem~\ref{theo:bound-leakage} for $v=2$), but 
 we  proved  that $f(u,\epsilon)$ is \rev{a bound, and therefore} the
optimal bound\footnote{When discussing our result with Barthe and
K\"opf, they said that they also conjectured that $f(u,\epsilon)$ is the
optimal bound.}.

Clarkson and Schneider also considered differential privacy as a case study of their proposal for quantification of integrity  \cite{Clarkson:11:TECHREP}. There, the authors analyzed database privacy conditions from the literature (such as differential privacy, $k$-anonymity, and $l$-diversity) using their framework for utility quantification. In particular, they studied the relationship between differential privacy and a notion of leakage (which is different from ours - in particular their definition is based on Shannon entropy) and they provided a tight bound on leakage. 

Heusser and Malacaria \cite{Heusser:09:FAST} were among the first to explore the application of information-theoretic concepts to databases queries. They proposed to model database queries as programs, which allows for statical analysis of the information leaked by the query.  However  \cite{Heusser:09:FAST}  did not attempt to relate information leakage to differential privacy.

In \cite{Ghosh:09:STC} the authors aimed at obtaining optimal-utility randomization mechanisms while preserving differential privacy. The authors proposed adding noise to the output of the query according to the geometric mechanism. Their framework is very interesting because it provides us with a general definition of utility for a randomization mechanism $M$ that captures any possible side information and preference (defined as a loss function) the users of $M$ may have. They proved that the geometric mechanism is optimal in the particular case of counting queries. Our results in Section \ref{sec:utility} do not restrict to counting queries, however we only consider the case of binary loss function.

%% file: conclusion.tex
An important question in statistical databases is how to deal with the trade-off between the privacy offered to the individuals participating in the database and the utility provided by the answers to the queries. In this work we proposed a model integrating the notions of privacy and utility in the scenario where differential-privacy is applied. We derived a strict bound on the information leakage of a randomized function satisfying $\epsilon$-differential privacy and, in addition, we studied the utility of oblivious differential privacy mechanisms. We provided a way to optimize utility while guaranteeing differential privacy, in the case where a binary gain function is used to measure the utility of the answer to a query. 

As future work, we plan to find bounds for more generic gain functions, possibly by using the Kantorovich metric to compare the a priori and a posteriori probability distributions on secrets.

%% file: appendix.tex
\subsection*{Notation}

In the following we assume that $A$ and $B$ are random variables with carriers $\mathcal{A}$ and $\mathcal{B}$, respectively. Let $M$ be a channel matrix with input $A$ and output $B$. We recall that the matrix $M$ represents the conditional probabilities $p_{B|A}(\cdot|\cdot)$. More precisely, the element of $M$ at the intersection of  row $a\in\mathcal{A}$ and  column $b\in\mathcal{B}$ is $M_{a,b} = p_{B|A}(b|a)$.
Note that if  the matrix  $M$ and  the input random variable $A$ are given,  then the output random variable $B$ is completely determined by them, and we use the notation $B(M,A)$ to represent this dependency. We also use   $H_{\infty}^{M}(A)$ to represent the conditional min-entropy $H_{\infty}(A|B(M,A))$. 
Similarly, we use $I_{\infty}^{M}(A)$ to denote $I_{\infty}(A;B(M,A))$. 

We denote by $M[l \to k]$  the matrix obtained by ``collapsing'' the column $l$ into $k$, i.e.
	\[
		M[l\to k]_{i,j} =
		\begin{cases}
			M_{i,k} + M_{i,l}	& j = k \\
			0					& j = l \\
			M_{i,j}				& \textrm{otherwise}
		\end{cases}
	\]

Given a partial function $\rho: \mathcal{A} \rightarrow \mathcal{B}$, the image of $\mathcal{A}$ under $\rho$ is $\rho(\mathcal{A}) = \{ \rho(a) | a \in \mathcal{A}, \rho(a) \neq \bot \}$, where $\bot$ stands for  ``undefined''. 

In the proofs we  need to use several indices, hence we  typically use the letters $i,j,h,k,l$  to range over rows and columns (usually $i,h,l$  range over rows and $j, k$  range over columns). 
 Given a matrix $M$, we denote by $\maxj{M}{j}$ the maximum value of column $j$ over all rows $i$, i.e. $\maxj{M}{j} = \max_{i}M_{i,j}$ . 

\subsection*{Proofs}

For the proofs, it will be useful to consider matrices with certain symmetries. In particular, it will be useful to transform our matrices in square matrices 
having the property that the  elements of the diagonal contain the   maximum values of each column, and are all equal.  This is the purpose of the following two lemmata: the first one transforms a matrix into a square matrix with all the column maxima in the diagonal, and the second makes all the elements of the diagonal equal. Both transformations preserve $\epsilon$-differential privacy and min-entropy leakage. 

\paragraph{\large Leakage}
\ \\[2ex]
\noindent 
In this part we prove the results about the bounds on min-entropy leakage. In the following lemmata, we assume that $M$ has input $A$ and output $B$, and that $A$ has a uniform distribution.

\begin{lemma}
	\label{lemma:collapse}
	Given an $n\times m$ channel matrix $M$ with $n \leq m$, providing $\epsilon$-differential
	privacy for some $\epsilon \geq 0$, we can construct a square $n \times n$ channel matrix $M'$
	such that:
	\begin{enumerate}
		\item $M'$ provides $\epsilon$-differential privacy.
	    \item $M'_{i,i} = \maxj{M'}{i}$ for all $i\in\cala$, i.e. the diagonal contains the maximum values of the columns.
		\item $H_\infty^{M'}(A) = H_\infty^{M}(A)$.
	\end{enumerate}
\end{lemma}

\begin{proof}
		We first show that there exists an $n\times m$ matrix $N$ and an injective total function
	$\rho: \mathcal{A} \rightarrow \mathcal{B}$ such that:
	\begin{itemize}
		\item $N_{i,\rho(i)} = 
		\maxj{N}{\rho(i)}$
		for all $i \in \mathcal{A}$,
		\item $N_{i,j} = 0$ for all $j \in \mathcal{B} \backslash \rho(\mathcal{A})$ and all $i \in \mathcal{A}$.
	\end{itemize}
	We iteratively construct $\rho,N$ ``column by column'' via a sequence of approximating partial functions $\rho_s$ and matrices $N_s$  ($0\leq s \leq  m$).
	
	\begin{itemize}
		\item \emph{Initial step} ($s = 0$).\\
			Define $\rho_0(i) = \bot$ for all $i \in \mathcal{A}$ and $N_0 = M$. \\	
		\item \emph{$s^{th}$ step} ($1 \leq s \leq m$). \\
			Let $j$ be the $s$-th column and let
			$i \in \mathcal{A}$ be one of the rows containing the maximum
			value of column $j$ in $M$, i.e.  $M_{i,j} = \maxj{M}{j}$.
			There are two cases:
				\begin{enumerate}
					\item $\rho_{s-1}(i) = \bot$: we define
						\begin{flalign*}
							\rho_{s} &= \rho_{s-1} \cup \{ i \mapsto j \} & \\
							N_s &= N_{s-1} &
						\end{flalign*}
					\item \label{item:b} $\rho_{s-1}(i) = k \in \mathcal{B}$: we define
						\begin{flalign*}
							\rho_s &= \rho_{s-1} & \\
							N_s &= N_{s-1}[j \to k] &
						\end{flalign*}
				\end{enumerate} 
	\end{itemize}

	Since the first step assigns $j$ in $\rho_s$ and the second zeroes the column $j$ in $N_s$, all unassigned
	columns $\calb \setminus \rho_m(\cala)$ must be zero in $N_m$. We finish the construction by taking $\rho$ to be the same
	as $\rho_m$ after assigning to each unassigned row one of the columns in  $\calb \setminus \rho_m(\cala)$ (there are enough such
	columns since $n \le m$). We also take $N = N_m$. Note that by construction
	$N$ is a channel matrix.

	Thus we get a matrix $N$ and a function $\rho: \mathcal{A} \rightarrow
	\mathcal{B}$ which, by construction, is injective and satisfies
	$N_{i,\rho(i)} = \maxj{N}{\rho(i)}$ for all $i \in \mathcal{A}$, and
	$N_{i,j} = 0$ for all $j \in \mathcal{B} \backslash \rho(\mathcal{A})$ and
	all $i \in \mathcal{A}$. Furthermore, $N$ provides $\epsilon$-differential
	privacy because each column is a linear combination of columns of $M$. It is also
	easy to see that $\sum_{j} \maxj{N}{j} = \sum_{j} \maxj{M}{j}$, hence $H_\infty^{N}(A) = H_\infty^{M}(A)$ (remember that A has the uniform distribution).

	Finally, we create our claimed matrix $M'$ from $N$ as follows: first, we
	eliminate all columns in $\calb\setminus\rho(\cala)$. Note that all these
	columns are zero so the resulting matrix is a proper channel matrix,
	provides differential privacy and has the same conditional min-entropy. Finally, we
	rearrange the columns according to $\rho$. Note that the order of the
	columns is irrelevant, any permutation represents the same conditional
	probabilities thus the same channel. The resulting matrix $M'$ is
	$n\times n$ and has all maxima in the diagonal.
\end{proof}

\begin{lemma}
	\label{lemma:vu}
	Let $M$ be a channel with input and output alphabets $\cala = \calb = {\it Val}^u$,
	and let $\sim$ be the adjacency relation on $Val^u$ defined  in Section~\ref{section:model}.
	Assume that the maximum value of each column is on the diagonal, that is $M_{i,i} = \maxj{M}{i}$ for all $i\in \cala$.
	If $M$ provides $\epsilon$-differential privacy then	we can construct a new channel matrix $M'$ such that:
	\begin{enumerate}
		\item $M'$ provides $\epsilon$-differential privacy;
		\item $M'_{i,i} = M'_{h,h}$ for all  $i,h \in \cala$ i.e. all the elements
		of the diagonal are equal;
		\item $M'_{i,i} = \maxj{M'}{i}$ for all $i\in \cala$;
		\item $H_\infty^{M}(A) = H_\infty^{M'}(A)$.		
	\end{enumerate}
\end{lemma}

\begin{proof}
	Let $k,l \in {\it Val}^u$. Recall that  $\mathit{dist}(k,l)$ (distance between $k$ and $l$) is   the length of the minimum $\sim$-path connecting $k$ and
	$l$ (Definition~\ref{def:distance-border}), i.e. the number of individuals in which
	$k$ and $l$ differ. Since $\cala=\calb={\it Val}^u$ we will use $\mathit{dist}(\cdot,\cdot)$ also between
	rows and columns. Recall also that ${\it Border}_{d}(h)= \{ k \in \calb | \mathit{dist}(h,k) = d \}$. 
	For typographical reasons,  in this proof we will use the   
	notation 
	$\calb_{h,d}$ to represent ${\it Border}_{d}(h)$, and $d(k,l)$ to represent $\mathit{dist}(k,l)$.

	Let $n = |\cala|=v^u$. The matrix $M'$ is given by
	\[
		M'_{h,k} = \frac{1}{n|\calb_{h,d(h,k)}|} \sum_{i\in\cala} 
			\sum_{j \in \calb_{i,d(h,k)}} M_{i,j}
	\]
	We first show that this is a well defined channel matrix, namely
	$ \sum_{k\in\calb} M'_{h,k}= 1$ for all $h \in \cala $. We have
	\begin{align*}
		\sum_{k\in\calb} M'_{h,k} &=
			\sum_{k\in\calb} 
				\frac{1}{n|\calb_{h,d(h,k)}|} \sum_{i\in\cala} \sum_{j \in \calb_{i,d(h,k)}} M_{i,j}
		\\ &=
				\frac{1}{n} \sum_{i\in\cala} \sum_{k\in\calb} \frac{1}{|\calb_{h,d(h,k)}|}\sum_{j \in \calb_{i,d(h,k)}} M_{i,j}
	\intertext{%
		Let $\Delta=\{0,\ldots,u\}$. Note that $\calb = \bigcup_{d\in \Delta}\calb_{h,d}$,
		and these sets are disjoint, so the summation over $k\in\calb$ can be split as follows
	}
		&=
				\frac{1}{n} \sum_{i\in\cala} 
					\sum_{d\in \Delta}\sum_{k\in\calb_{h,d}}  \frac{1}{|\calb_{h,d}|}\sum_{j \in \calb_{i,d}} M_{i,j}
		\\&=
				\frac{1}{n} \sum_{i\in\cala} 
					\sum_{d\in \Delta} \sum_{j \in \calb_{i,d}} M_{i,j} \sum_{k\in\calb_{h,d}}\frac{1}{|\calb_{h,d}|}
	\intertext{%
		as  $\sum_{k\in\calb_{h,d}}\frac{1}{|\calb_{h,d}|} =1$, we obtain 
	}
		&=
				\frac{1}{n} \sum_{i\in\cala} 
					\sum_{d\in \Delta} \sum_{j \in \calb_{i,d}} M_{i,j} 
	\intertext{%
		and now the summations over $j$ can be joined together
	}
		&=
				\frac{1}{n} \sum_{i\in\cala} 
					\sum_{j \in \calb} M_{i,j}
		= 1
	\end{align*}
	We now show that the elements of the diagonal have the intended properties.
	First, we show that the elements of the diagonal are all the same.
	We have that $\calb_{i,d(h,h)} = \calb_{i,0} =\{i\}$ for all $h\in\cala$, and therefore:
	\[
		M'_{h,h} = \frac{1}{n}\sum_{i\in \cala}M_{i,i}
	\]
	Then, we show that they are the maxima for each
	column. Note that $|\calb_{i,d}| = \binom{u}{d} (v-1)^{d}$ which is
	independent of $i$. We have:
	\begin{align*}
		M'_{h,k} 
			&= \frac{1}{n|\calb_{h,d(h,k)}|} \sum_{i\in\cala} \sum_{j \in \calb_{i,d(h,k)}} M_{i,j} \\
			&\le \frac{1}{n|\calb_{h,d(h,k)}|} \sum_{i\in\cala} \sum_{j \in \calb_{i,d(h,k)}} M_{i,i}
				& \textrm{($M$ has maxima in the diag.)} \\
			&= \frac{1}{n} \sum_{i\in\cala} \frac{|\calb_{i,d(h,k)}|}{|\calb_{h,d(h,k)}|} M_{i,i} \\
			&= \frac{1}{n} \sum_{i\in\cala} M_{i,i} = M'_{h,h}
	\end{align*}
	It easily follows that $\sum_{j} \maxj{M'}{j} = \sum_{j} \maxj{M}{j}$
	which implies that $H_\infty^{M}(A) = H_\infty^{M'}(A)$.

	It remains to show that $M'$ provides $\epsilon$-differential privacy, namely that
	\[
		M'_{h,k} \le e^\epsilon M'_{h',k} \qquad \forall h,h',k \in \cala :
		h \sim h'
	\]
	Since $d(h,h')=1$, by the triangular inequality we derive:
	\[
		d(h',k) - 1 \le d(h,k) \le d(h',k) + 1
	\]
	Thus, there are exactly 3 possible cases: 
	
	\begin{enumerate}
	
		\item $d(h,k)=d(h',k)$.\\
		The result is immediate since $M'_{h,k}=M'_{h',k}$. \\
		
		\item $d(h,k)=d(h',k)-1$.\\
		Define
		\[
			\cals_{i,j} = \{ j' \in \calb_{i,d(i,j)+1} | j' \sim j \}
		\]
		Note that $|\cals_{i,j}| = (u - d(i,j))(v-1)$ ($i$ and $j$ are equal in
		$u-d(i,j)$ elements, and we can change any of them in $v-1$ ways). The following
		holds:
		\begin{align}
			M_{i,j} &\le e^\epsilon M_{i,j'}	\qquad \forall j' \in \cals_{i,j}
				&&\text{(diff. privacy)}\Rightarrow   \notag\\
			(u - d(i,j))(v-1)M_{i,j} &\le e^\epsilon \sum_{j'\in \cals_{i,j}} M_{i,j'}
				&&\text{(sum of the above)}\Rightarrow \notag\\
			\sum_{j\in \calb_{i,d(h,k)}} (u - d(h,k))(v-1)M_{i,j}
				&\le  e^\epsilon \sum_{j\in \calb_{i,d(h,k)}} \sum_{j'\in \cals_{i,j}} M_{i,j'}
				&&\text{(sum over $j$)} \notag
		\intertext{%
			Let $d=d(h,k)$.
			Note that each $j' \in \calb_{i,d+1}$ is contained in exactly $d+1$
			different sets $\cals_{i,j}, j \in \calb_{i,d}$. So the right-hand side
			above sums all elements of $\calb_{i,d+1}$, $d+1$ times each. Thus
			we get
		}
			\label{eq1}
			(u - d)(v-1)\sum_{j\in \calb_{i,d}} M_{i,j}
				&\le e^\epsilon\ (d+1) \sum_{j\in \calb_{i,d+1}} M_{i,j}
		\end{align}
		Finally, we have
		\begin{align*}
			M'_{h,k} &= \frac{1}{n|\calb_{h,d}|} \sum_{i\in\cala} 
				\sum_{j \in \calb_{i,d}} M_{i,j}
					\\
			&\le e^\epsilon\ 
			\frac{1}{n\binom{u}{d}(v-1)^{d}}   
			\frac{d+1}{(u - d)(v-1)}
			\sum_{i\in\cala} 
				\sum_{j \in \calb_{i,d+1}} M_{i,j}
					&& \text{(from \eqref{eq1})}
					\\
			&\le 
			e^\epsilon\ 
			\frac{1}{n\binom{u}{d+1}(v-1)^{d+1}}    \sum_{i\in\cala} 
				\sum_{j \in \calb_{i,d+1}} M_{i,j}
				\\
			&= e^\epsilon M'_{h',k}
		\end{align*}\\
		
		\item $d(h,k)=d(h',k)+1$.\\
		Symmetrical to the case $d(h,k)=d(h',k)-1$.\\
	
	\end{enumerate}
	
\end{proof}

We are now ready to prove our first main result. 

\begin{theo} {\bf \ref{theo:bound-leakage}.}
	If $\mathcal{K}$ provides \emph{$\epsilon$-differential privacy} then the min-entropy leakage associated to $\mathcal{K}$ is bounded from above as follows:
	\[I_\infty(X;Z)  \leq u\, \log_2\frac{v\, e^\epsilon}{(v-1 +  e^\epsilon)}\]
\end{theo}

\begin{proof}
		
	Let us assume, without loss of generality, that $|{\cal X}| \leq |{\cal Z}|$  (if this is not the case, then we add enough zero columns, i.e. columns containing only $0$'s, so to match the number of rows. Note that adding zero columns does not change the min-entropy leakage). 

 For our proof we need a square matrix with all column maxima on the diagonal, and all equal. We obtain such a matrix by transforming the matrix associated to $\mathcal{K}$ as follows: first we  apply Lemma~\ref{lemma:collapse} to it (with $A=X$ and $B=Z$), and then we apply 
	Lemma~\ref{lemma:vu} to the result of  Lemma~\ref{lemma:collapse}. 
	The final matrix  $M$ has size $n\times n$, with $n = |{\cal X}| = v^u$,  provides $\epsilon$-differential privacy,  and for all rows $i,h$ 
	we have that $M_{i,i} =  M_{h,h}$ and $M_{i,i} = \maxj{M}{i}$. 
	  Furthermore, $I_\infty^M(X)$ is equal to the min-entropy leakage of $\mathcal{K}$. 

Let us denote by $\alpha$ the value of every element in the diagonal of $M$, i.e. $\alpha = M_{i,i}$ for every row $i$. Note that for every 
$j \in {\it Border}_{d}(i)$ (i.e. every $j$ at distance   $d$ from a given $i$) the  value of $M_{i,j}$ is at least $\frac{M_{i,i}}{(e^{ \epsilon})^d}$, hence 
$M_{i,j}\geq \frac{\alpha}{(e^{ \epsilon})^d}$. Furthermore each element $j$ at distance $d$ from $i$ can be obtained by changing the value  of $d$ individuals  in the $u$-tuple representing $i$. 
We can choose those $d$ individuals in $\binom{u}{d}$ possible ways, and for each of these individuals we can change the value (with respect to the one in $i$) in  $v-1$ possible ways. 
Therefore $|{\it Border}_{d}(i)| = \binom{u}{d} (v-1)^{d}$,  and we obtain:
\[
\displaystyle  \sum_{d=0}^{u} \binom{u}{d} (v-1)^{d} \frac{\alpha}{(e^{ \epsilon})^d} \, \leq  \, \sum_{j=1}^{n} M_{i,j} 
\]
Since each row represents a probability distribution, the elements of row $i$ must sum up to $1$. Hence:
\[
\displaystyle \sum_{d=0}^{u} \binom{u}{d} (v-1)^{d} \frac{\alpha}{(e^{ \epsilon})^d} \, \leq \,   1
\]
Now we apply some transformations:
	\begin{equation*}		
				\begin{array}{rcll}
				\displaystyle \sum_{d=0}^{u} \binom{u}{d} (v-1)^{d} \frac{\alpha}{(e^{ \epsilon})^d} & \;\leq \;& 1 & \; \iff \\
				\displaystyle \alpha \sum_{d=0}^{u} \binom{u}{d} (v-1)^{d} ((e^{ \epsilon})^d)^{u-d}& \; \leq\;& ({e^\epsilon})^{u} \\
			\end{array}
	\end{equation*}
Since $\displaystyle \alpha \sum_{d=0}^{u} \binom{u}{d} (v-1)^{d} ({e^\epsilon})^{u-d}\, = \, (v-1+e^\epsilon)^u$ (binomial expansion), we obtain:
	\begin{equation}		
		\label{eq:theorem:bl}
				\displaystyle \alpha \leq \left(\frac{{e^\epsilon}}{v-1+e^\epsilon}\right)^u
	\end{equation}
Therefore:	
	\begin{equation*}		
		\label{eq:}
		\begin{array}{lclll}
			\displaystyle I_{\infty}^{M}(X) & \,= \,&H_{\infty}(X) - H_{\infty}^{M}(X) & \quad&  \text{(by definition)} \\[2ex]
					& \,= \,&\displaystyle\log_{2} v^{u} + \log_2 \sum_j \alpha \frac{1}{n} &\\[2ex]
					& \,= \,&\log_{2} v^{u} + \log_2 \alpha  &\\[2ex]
					& \,\leq \,&\displaystyle\log_{2} v^{u} + \log_2 \left(\frac{{e^\epsilon}}{v-1+e^\epsilon}\right)^u  && \text{(by \eqref{eq:theorem:bl} )} \\[2ex]
					& \,= \,&\displaystyle u \log_{2} \frac{v\,{e^\epsilon}}{v-1+e^\epsilon}  & \\
		\end{array}
	\end{equation*}
	
\end{proof}
The next proposition shows that the bound obtained in previous theorem is tight.
\begin{prop} {\bf \ref{prop:tight}.}
	For every $u$, $v$, and $\epsilon$ there exists a  randomized function $\mathcal K$ 
	which provides $\epsilon$-differential privacy and whose min-entropy leakage, for the uniform input distribution, is $I_\infty(X;Z)=B(u,v,\epsilon)$.
\end{prop}

\begin{proof}
The adjacency relation in $\cal X$ determines a graph structure $G_{\cal X}$. 
Set $\cal{Z}=\cal{X}$ and define the matrix of $\mathcal K$  as follows: 
\[ p_{\mathcal K}(z|x) = \frac{B(u,v,\epsilon)}{(e^\epsilon)^d} \qquad\mbox{where $d$ is the distance between $x$ and $z$ in $G_{\cal X}$} \] 
It is easy to see that $p_{\mathcal{K}}(\cdot |x)$ is a probability distribution  for every $x$, that $\mathcal K$ provides $\epsilon$-differential privacy, and that  $I_\infty(X;Z)=B(u,v,\epsilon)$.
\end{proof}



We consider now the case in which $|\mathit{Range}({\mathcal{K})}|$ is bounded by a number smaller than $v^u$. 

In the following when we have a random variable $X$, and a matrix $M$ with row indices in $\cala\subsetneq\calx$,  we will use the notations  
$H_\infty^{M}(X)$ and $I_\infty^{M}(X)$ to represent the conditional min-entropy and leakage obtained by adding  ``dummy raws'' to $M$, namely rows that extend the input domain  	of the corresponding channel so to match the input   $X$,  but which do not contribute to the computation of $H_\infty^{M'}(X)$. Note that it is easy to extend $M$ this way: we only have to make sure that  for each  column $j$ the value of each of these new rows is dominated by $\maxj{M'}{j}$. 

We will also use the notation $\sim_u$ and $\sim_\ell$ to refer to the standard adjacency relations on $\mathit{Val}^u$ and  $\mathit{Val}^\ell$, respectively.

\begin{lemma}
	\label{lemma:new-bound}
	Let $\mathcal K$ be a randomized function with input $X$, where $\mathcal{X}=\mathit{Val}^\mathit{Ind}$, providing $\epsilon$-differential privacy.
	Asssume that $r = |\mathit{Range}({\mathcal{K})}|=v^\ell$, for some $\ell<u$. 
	Let $M$ be the matrix associated to $\mathcal K$. Then it is possible 
	to build a square matrix   $M'$  of size $v^\ell\times v^\ell$, with row and column indices in $\cala\subseteq\calx$,
	and a binary relation $\sim' \subseteq\cala\times \cala$ such that $(\cala, \sim')$ is isomorphic to $(\mathit{Val}^\ell, \sim_\ell)$, and such that:
	\begin{enumerate}
		\item $M'_{i,j}\leq (e^\epsilon)^{u-l + d}\, M'_{i,j}$ for all $i,j,k\in \cala$, where $d$ is the $\sim'$-distance between $j$ and $k$.
	         \item $M'_{i,i} = M'_{h,h}$ for all $i,h\in\cala$, i.e. elements of the diagonal are all equal
	         \item $M'_{i,i} = \maxj{M'}{i}$ for all $i\in\cala$, i.e. the diagonal contains the maximum values of the  columns.
                   \item $H_\infty^{M'}(X) = H_\infty^{M}(X)$. 
	\end{enumerate}
\end{lemma}

\begin{proof}
We first apply a procedure similar to that of Lemma~\ref{lemma:collapse} to construct a square matrix of size $v^\ell\times v^\ell$ which has the maximum values of each column in the diagonal. (In this case we construct an injection from the columns to rows containing their maximum value, and we eliminate the rows that at the end are not associated to any column.) Then  define $\sim'$ as the projection of $\sim_u$ on $\mathit{Val}^\ell$. Note that point $1$ in this lemma is satisfied by this definition of $\sim'$. 
Finally, apply the procedure in Lemma~\ref{lemma:vu} (on the structure  $(\cala, \sim')$) to make all elements in the diagonal equal and maximal. Note that this procedure preserves the 
property in point $1$, and conditional min-entropy. Hence $H_\infty^{M'}(X) = H_\infty^{M}(X)$. 
\end{proof}

\begin{prop}{\bf \ref{prop:new-bound}.}
Let $\mathcal K$ be a randomized function and let $r = |\mathit{Range}({\mathcal{K})}|$.  
If $\mathcal{K}$ provides \emph{$\epsilon$-differential privacy} then 
the min-entropy leakage associated to $\mathcal{K}$ is bounded from above as follows:
\[I_\infty(X;Z) \, \leq \,\log_2\frac{r\,(e^{\epsilon})^u}{(v-1+e^\epsilon)^\ell-(e^{\epsilon})^\ell+(e^{\epsilon})^u}\]
where $\ell=\lfloor \log_v r\rfloor$.
\end{prop}

\begin{proof}
Assume  first  that $r$ is of the form $v^\ell$. We transform the matrix $M$ associated to ${\mathcal{K}}$ by applying Lemma~\ref{lemma:new-bound}, and let $M'$ be the resulting matrix. 
Let us denote by $\alpha$ the value of every element in the diagonal of $M'$, i.e. $\alpha = M'_{i,i}$ for every row $i$, and let us denote by 
$ {\it Border'}_{d}(i)$ the border (Def~\ref{def:distance-border}) wrt $\sim'$. 
 Note that for every 
$j \in {\it Border'}_{d}(i)$  we have that
$ M'_{i,i}\leq {M'_{i,j}}{(e^{ \epsilon})^{u-\ell+d}}$,
 hence 
\[M'_{i,j}\leq \frac{\alpha}{(e^{ \epsilon})^{u-\ell+d}}\]
Furthermore each element $j$ at $\sim'$-distance $d$ from $i$ can be obtained by changing the value  of $d$ individuals  in the $\ell$-tuple representing $i$ (remember that $(\cala,\sim')$ is isomorphic to $(\mathit{Val}^\ell, sim_\ell)$). 
We can choose those $d$ individuals in $\binom{\ell}{d}$ possible ways, and for each of these individuals we can change the value (with respect to the one in $i$) in  $v-1$ possible ways. 
Therefore 
\[|{\it Border'}_{d}(i)| = \binom{\ell}{d} (v-1)^{d}\] 
Taking into account that for $M'_{i,i}$ we do not need to divide by  $(e^{ \epsilon})^{u-\ell+d}$, we obtain:
\[
\displaystyle  \alpha +\sum_{d=1}^{\ell} \binom{\ell}{d} (v-1)^{d} \frac{\alpha}{(e^{ \epsilon})^{u-\ell+d}} \, \leq  \, \sum_{j} M'_{i,j} 
\]
Since each row represents a probability distribution, the elements of row $i$ must sum up to $1$. Hence:
\begin{equation}
\label{eq:bl1}
\displaystyle \alpha + \sum_{d=1}^{u} \binom{u}{d} (v-1)^{d} \frac{\alpha}{(e^{ \epsilon})^{u-\ell+d}} \, \leq \,   1
\end{equation}
By performing  some simple calculations, similar to those of the proof of Theorem~\ref{theo:bound-leakage}, we obtain:
\[\alpha \, \leq \,\frac{(e^{\epsilon})^u}{(v-1+e^\epsilon)^\ell-(e^{\epsilon})^\ell+(e^{\epsilon})^u}\]
Therefore:	
	\begin{equation}		
		\label{eq:bl2}
		\begin{array}{lclll}
			\displaystyle I_{\infty}^{M'}(X) & \,= \,&H_{\infty}(X) - H_{\infty}^{M'}(X) & \quad&  \text{(by definition)} \\[3ex]
					& \,= \,&\displaystyle \log_{2} v^{u} + \log_2 \sum_{j=1}^{v^\ell} \alpha \frac{1}{v^u} &\\[3ex]
					& \,= \,&\displaystyle\log_{2} v^{u} + \log_2\frac{1}{v^u} + \log_2 (v^\ell\,\alpha)  &\\[3ex]
					& \,\leq \,&\displaystyle \log_2 \frac{v^\ell\,(e^{\epsilon})^u}{(v-1+e^\epsilon)^\ell-(e^{\epsilon})^\ell+(e^{\epsilon})^u}  && \text{(by \eqref{eq:bl1} )}
		\end{array}
	\end{equation}

Consider now the case in which $r$ is not of the form $v^\ell$. Let $\ell$ be the maximum integer such that $v^\ell < r$, and let $m = r-v^\ell$.  We transform the matrix $M$ associated to 
${\mathcal{K}}$ by collapsing the $m$ columns with the  smallest maxima into the $m$ columns with highest maxima. Namely, let $j_1,j_2,\ldots,j_m$ the indices of the columns which  have the
smallest maxima values, i.e. $\maxj{M}{j_t}\leq\maxj{M}{j}$ for every column $j\neq j_1,j_2,\ldots,j_m$. Similarly, let $k_1,k_2,\ldots,k_m$ be the indexes of the columns which have maxima values. 
Then, define
\[
N  =  M[j_1\rightarrow k_1] 	[j_2\rightarrow k_2] 	\ldots [j_m\rightarrow k_m] 		
\]	 
Finally, eliminate the $m$ zero-ed columns to obtain a   matrix with exactly  $v^\ell$ columns. 		
It is easy to show that 
\[
I_\infty^M(X)\, \leq \, I_\infty^N(X) \frac{r}{v^\ell}
\]
After transforming $N$ into a matrix $M'$ with the same min-entropy leakage as described in the first part of this proof,  from (\ref{eq:bl2}) we conclude
\[
I_\infty^M(X)\, \leq \, I_\infty^{M'}(X) \frac{r}{v^\ell} \,\leq \, \displaystyle \log_2 \frac{r\,(e^{\epsilon})^u}{(v-1+e^\epsilon)^\ell-(e^{\epsilon})^\ell+(e^{\epsilon})^u} 
\]
\end{proof}

We now turn our attention to the min-entropy leakage associated to an individual. 

\begin{lemma}
	\label{lemma:Renyi}
	If a randomized function $\mathcal{K}: A \rightarrow B$ respects an $\epsilon$-ratio in the sense that $p_{\mathcal{K}}(b|a') \leq e^{\epsilon} \cdot p_{\mathcal{K}}(b|a'')$ for all $a', a'' \in \mathcal{A}$ and $b \in \mathcal{B}$, then the min-entropy leakage from $A$ to $B$ is bounded by:
	$$I_{\infty}(A;B) \leq \epsilon \log_2 e$$
\end{lemma}

\begin{proof}

	For clarity reasons, in this proof we use the notation $p(b|A=a)$ for the probability distributions $p_{\mathcal{K}}(b|A=a)$ associated to $\mathcal{K}$.

	\begin{equation*}		
		\label{eq:renyi-cond-entropy-aux}
		\begin{array}{lcll}
			 - H_{\infty}(A|B) & =    &  \log_2 \sum_{b} p(b) \max_{a} p(a|b) &\quad \mbox{(by definition)} \\[2ex]
			                                & =    &  \log_2 \sum_{b} \max_{a} (p(b)\, p(a|b)) &\quad \mbox{} \\[2ex]
			                                & =    &  \log_2 \sum_{b} \max_{a} (p(a)\, p(b|a) )&\quad \mbox{(by the Bayes theorem)} \\[2ex]
			                                & \leq &  \log_2 \sum_{b} \max_{a} ( p(a)\, e^{\epsilon}\, p(b|\hat{a}) )&\quad \mbox{(by hypothesis on $\mathcal{K}$, for some fixed $\hat{a}$)} \\[2ex]
			                                & =    &  \log_2 \sum_{b} e^{\epsilon} \, p(b|\hat{a}) \max_{a} p(a) &\quad \mbox{} \\[2ex]
			                                & =    &  \log_2 \left( e^{\epsilon} \max_{a} p(a) \sum_{b} p(b|\hat{a})  \right) &\quad \mbox{} \\[2ex]
			                                & =    &  \log_2 \left( e^{\epsilon} \max_{a} p(a) \right) &\quad \mbox{(by probability laws)} \\[2ex]
			                                & =    &  \log_2 e^\epsilon + \log \max_{a} p(a)   & \quad\mbox{} \\[2ex]
			                                & =    &  \epsilon  \log_2 e - H_{\infty}(A)       & \quad\mbox{(by definition)} 
		\end{array}
	\end{equation*}
	
	Therefore:
	
	\begin{equation}
		\label{eq:renyi-cond-entropy}
		H_{\infty}(A|B) \geq H_{\infty(A)} - \epsilon \log_2 e
	\end{equation}
	
	This gives us a bound on  the min-entropy leakage:
	
	\begin{equation*}		
		\label{eq:renyi-mutual-info}
		\begin{array}{lcll}
			\displaystyle I_{\infty}(A;B) & =    & \displaystyle H_{\infty}(A) - H_{\infty}(A|B) &   \\[2ex]
			                              & \leq & \displaystyle \epsilon \log_2 e & \quad\mbox{(by~(\ref{eq:renyi-cond-entropy}))}  
		\end{array}
	\end{equation*}
	
\end{proof}

\begin{theo}{\bf \ref{theo:Renyi}.}
	If $\mathcal{K}$ provides \emph{$\epsilon$-differential privacy} then for all $D^- \in {\it Val}^{u-1}$ the min-entropy leakage about an individual is bounded from above as follows:
	\[ I_{\infty}(X_{D^-};Z) \leq \log_2 e^\epsilon \]
\end{theo}

\begin{proof}
	By construction, the elements of $\mathcal{X}_{D^-}$ are all adjacent. Hence $\mathcal{K}_{D^-}$ respects an $\epsilon$-ratio. Thus we are allowed  to apply  Lemma~\ref{lemma:Renyi} (with $X=X_{D^-}$ and $\mathcal{K}=\mathcal{K}_{D^-}$), which gives immediately the intended result.
\end{proof}

\paragraph{\large Utility} 
\ \\[2ex]
In this part we prove the results on utility. We start with a lemma which plays a role analogous to Lemma~\ref{lemma:vu}, but for a different kind of graph structure: in this case, we require the graph to have an automorphism with a single orbit. 

\begin{lemma}
	\label{lemma:automorphism} 
	Let $M$ be the matrix of a channel with the same input and output alphabet $\cala$.
	Assume an adjacency relation $\sim$ on $\cala$ such that the graph $(\mathcal{A}, \sim)$
	has an automorphism $\sigma$ with a single orbit.
	Assume that the maximum value of each column is on the diagonal, that is $M_{i,i} = \maxj{M}{i}$ for all $i\in \cala$.
	If $M$ provides $\epsilon$-differential privacy then	we can construct a new channel matrix $M'$ such that:
	\begin{enumerate}
		\item $M'$ provides $\epsilon$-differential privacy;
		\item $M'_{i,i} = M'_{h,h}$ for all  $i,h \in \cala$;
		\item $M'_{i,i} = \maxj{M'}{i}$ for all $i\in \cala$;
		\item $H_\infty^{M}(A) = H_\infty^{M'}(A)$.		
	\end{enumerate}
\end{lemma}


\begin{proof}
	Let $n=|\cala|$. For every $h,k\in\cala$ let us define the elements of $M'$ as:
	
	\begin{equation*}
		M'_{h,k} = \frac{1}{n} \sum_{i=0}^{n-1}M_{\sigma^{i}(h),\sigma^{i}(k)}
	\end{equation*}
	First we prove that $M'$ provides $\epsilon$-differential privacy. For every pair $h \sim l$ and every $k$:
	\begin{equation*}
		\begin{array}{lcll}
			M'_{h,k} & =    & \displaystyle \sum_{i=0}^{n-1} M_{\sigma^{i}(h),\sigma^{i}(k)} & \\[3ex]
							& \leq & \displaystyle \sum_{i=0}^{n-1} e^\epsilon M_{\sigma^{i}(l),\sigma^i(k)}& \qquad\text{(by $\epsilon$-diff. privacy, for some $l$ s.t. $\rho(\sigma^{i}(h')) = k$)} \\[3ex]
							& =    & e^\epsilon M'_{l,k} &\\
		\end{array}
	\end{equation*}
	Now we prove that for every $h$, $M'_{h,\cdot}$ is a legal probability distribution. Remember that $\{\sigma^i(k)|0\leq i\leq n-1\}=\cala$ since $\sigma$ has a single orbit.
	\begin{equation*}
		\begin{array}{lcll}
			\displaystyle \sum_{k=0}^{n-1} M'_{h,k} 
				& = & \displaystyle \sum_{k=0}^{n-1} \frac{1}{n} \sum_{i=0}^{n-1}M_{\sigma^{i}(h),\sigma^{i}(k)}, &\\[3ex]
				& = & \displaystyle \sum_{i=0}^{n-1} \frac{1}{n} \sum_{k=0}^{n-1}M_{\sigma^{i}(h) ,\sigma^{i}(k)} & \\[3ex]
				& = & \displaystyle \sum_{i=0}^{n-1} \frac{1}{n} \,  1 &\quad\text{(since $\{\sigma^i(k)|0\leq i\leq n-1\}=\cala$ )} \\[3ex]
				& = & \displaystyle 1 & \\
		\end{array}
	\end{equation*}
	Next we prove that the diagonal contains the maximum value of each column, i.e., for every $k$, $M'_{k,k} = \maxj{M'}{k}$.
	\begin{equation*}
		\begin{array}{lcll}
			\displaystyle M'_{k,k} & =    & \displaystyle \frac{1}{n} \sum_{k=0}^{n-1} M_{\sigma^{i}(k),\sigma^{i}(k)} & \\[3ex]
														& \geq & \displaystyle \frac{1}{n} \sum_{k=0}^{n-1} M_{\sigma^{i}(h),\sigma^{i}(k)} \quad\text{{(since $M_{\sigma^{i}(h),\sigma^{i}(h)}=\textstyle \maxj{M}{\sigma^i(h)}$)}}&\\[3ex]
														& =    & \displaystyle M'_{hk} &\\
		\end{array}
	\end{equation*}
	Finally, we prove that $I_\infty^{M'}(A) = I_\infty^{M}(A)$. It is enough to prove that $H_\infty^{M'}(A) = H_\infty^{M}(A)$.
	\begin{equation*}
		\begin{array}{lcll}
			\displaystyle H_\infty^{M'}(A) & =    & \displaystyle \sum_{h=0}^{n-1} M_{h,h} & \\[3ex]
					&= & \displaystyle\frac{1}{n} \sum_{h=0}^{n-1}\sum_{i=0}^{n-1}   M_{\sigma^{i}(h),\sigma^{i}(h)} & \quad\text{(since $\{\sigma^i(h)|0\leq i\leq n-1\}=\cala$)} \\[3ex]
				         & = & \displaystyle \frac{1}{n} \sum_{h=0}^{n-1}  H_\infty^{M}(A) &  \quad\text{(since $M_{\sigma^{i}(h),\sigma^{i}(h)}=\maxj{M}{\sigma^i(h)}$)} \\[3ex]
					& = & \displaystyle H_\infty^{M}(A) &\\
		\end{array}
	\end{equation*}

\end{proof}

\begin{theo} {\bf \ref{theorem:utility-bound}.}
Let $\cal H$ be a randomization mechanism for the randomized function $\cal K$ and the query $f$, and assume that $\cal K$ provides $\epsilon$-differential privacy. 
Assume that $({\cal Y},\sim)$ admits a graph automorphism with a single orbit. Furthermore, assume that there exists a natural number $c$ and an element $y\in{\cal Y}$ such that, \rev{for every natural number $d>0$}, either $|{\it Border}_{d}(y)| = 0$ or $|{\it Border}_{d}(y)| \geq c$. Then
\[
\calu(X,Y) \leq \frac{(e^\epsilon)^n(1-e^\epsilon)}{(e^\epsilon)^n(1-e^\epsilon) + c\,(1-(e^\epsilon)^n)}
\]
where $n$ is the maximum distance from $y$ in $\cal Y$.
\end{theo}

\begin{proof}

Consider the matrix $M$  obtained by applying Lemma~\ref{lemma:collapse} to the matrix of  $\cal H$, and then Lemma~\ref{lemma:automorphism} to the result of Lemma ~\ref{lemma:collapse}. Let us call $\alpha$ the value of the elements in the diagonal of $M$.

	Let us take an element $M_{i,i} = \alpha$. For each element $j\in{\it Border}_{d}(M_{i,i})$, the value of $M_{i,j}$  can be at most $\frac{\alpha}{e^{d \epsilon}}$. Also, the elements of row $i$ represent a probability distribution, so they sum up to 1. Hence we obtain:
\[\displaystyle \alpha + \sum_{d=1}^{n} |{\it Border}_{d}(y)| \frac{\alpha}{(e^\epsilon)^d}\; \leq\; 1\]
Now we perform some simple calculations:	
	\begin{equation*}		
		\label{eq:theorem:bl2}
			\begin{split}
				\displaystyle \alpha + \sum_{d=1}^{n} |{\it Border}_{d}(y)| \frac{\alpha}{(e^\epsilon)^d}\; \leq\; 1 & \implies \text{(since by hypothesis $ |{\it Border}(y,d)| \geq c$)} \\		
				\displaystyle \alpha + \sum_{d=1}^{n} c\, \frac{\alpha}{(e^\epsilon)^d} \;\leq \;1 & \iff \text{} \\		
				\displaystyle \alpha \,(e^\epsilon)^n+ c \, \alpha\, \sum_{d=1}^{n} {(e^\epsilon)^{n-d}} \leq (e^\epsilon)^n & \iff \text{} \\	
				\displaystyle \alpha \,(e^\epsilon)^n+ c\,\alpha\, \sum_{t=0}^{n-1} {(e^\epsilon)^t} \leq (e^\epsilon)^n & \iff \text{(geometric progression sum)} \\	
				\displaystyle \alpha \,(e^\epsilon)^n+ c\, \alpha\, \frac{ 1-(e^{\epsilon})^n}{1-e^{\epsilon}} \leq (e^\epsilon)^n & \iff \text{} \\
				\displaystyle \alpha \leq \frac{(e^{\epsilon})^n (1-e^{\epsilon})}{(e^{\epsilon})^n(1-e^{\epsilon}) + c\,(1-(e^{\epsilon})^n)} & \text{} \\		
			\end{split}
	\end{equation*}
	Since $\mathcal{U}(Y,Z) = \alpha$,  we conclude.	

\end{proof}

\begin{theo}{\bf \ref{theorem:construction-h}.}
Let $f:\calx\rightarrow \caly$ be a query and  let $\epsilon\geq 0$. Assume that $({\cal Y},\sim)$  admits a graph automorphism with a single orbit, and that there exists $c$ such that, for every $y\in \caly$ and \rev{every natural number $d>0$}, either $|{\it Border}_{d}(y)| = 0$ or $|{\it Border}_{d}(y)| = c$. Then, for such $c$,  the definition  in (\ref{eqn:optimal-matrix}) determines a legal channel matrix for $\cal H$, i.e.,  for each $y \in \caly$, $p_{Z|Y}(\cdot|y)$ is a probability distribution. Furthermore,   the composition  $\cal K$ of $f$ and $\cal H$ provides $\epsilon$-differential privacy. Finally, $\cal H$ is optimal in the sense that it maximizes utility when the  distribution of $Y$ is uniform. 
\end{theo}

\begin{proof}
We follow a reasoning analogous to the proof of Theorem~\ref{theorem:utility-bound}, but using $ |{\it Border}(y,d)| = c$, to prove that   
\[
 \mathcal{U}(Y,Z) = \frac{(e^{\epsilon})^n (1-e^{\epsilon})}{(e^{\epsilon})^n(1-e^{\epsilon}) + c\,(1-(e^{\epsilon})^n)}
\]
 From the same theorem, we know that this is a maximum for the utility. 
\end{proof}